\documentclass[12pt]{article}
\usepackage{graphicx}
\usepackage{amssymb, amsmath, amsthm}

\newtheorem{theorem}{Theorem}
\newtheorem{lemma}[theorem]{Lemma}

\begin{document}

\title{About the k-Error Linear Complexity over $\mathbb{F}_p$ of
 sequences of length 2$p$ with optimal three-level
autocorrelation}
\author{Vladimir Edemskiy}

\maketitle

\begin{abstract}
We investigate the $k$-error linear complexity over $\mathbb{F}_p$
of binary sequences of length $2p$ with optimal three-level
autocorrelation. These balanced sequences  are constructed by
cyclotomic classes of order four using a method presented by  Ding
et al.



\end{abstract}

\section{Introduction}

\label{in} Autocorrelation is an important measure of pseudo-random
sequence for their application in code-division multiple access
systems, spread spectrum communication systems, radar systems and so
on \cite{GG}. An important problem in sequence design is to find
sequences with optimal autocorrelation. In their paper, Ding et al.
\cite{DHM} gave several new families of binary sequences of period
$2p$ with optimal autocorrelation $\{-2,2\}$.

The linear complexity is another important characteristic of
pseudo-random sequence,  which is significant  for cryptographic
applications. It is defined as the length of the shortest linear
feedback shift register that can generate the sequence \cite{LN}.
The linear complexity  of above-mention sequences over the finite
field of order two was investigated in \cite{ZZ} and in \cite{EP}
over the finite field $\mathbb{F}_p$ of $p$ elements and other
finite fields. However, high linear complexity can not guarantee
that the sequence is secure. For example, if changing one or few
terms of a sequence can greatly reduce its linear complexity, then
the resulting key stream would be cryptographically weak. Ding et
al. \cite{DXS} noticed this problem first in their book, and
proposed the weight complexity and the sphere complexity. Stamp and
Martin \cite{SM} introduced the $k$-error linear complexity, which
is the minimum of the linear complexity and sphere complexity.  The
$k$-error linear complexity of a sequence $r$ is defined by $L_k(r)
=
  \min \limits_{t} L(t),$ where the minimum of the linear complexity $L(t)$ is taken over all
$N$-periodic sequences $t =(t_n)$ over $\mathbb{F}_p$ for which the
Hamming distance of the vectors $(r_0, r_1, \dots , r_{N-1})$ and
$(t_0, t_1, \dots , t_{N-1})$ is at most $k$. Complexity measures
for sequences over finite fields, such as the linear complexity and
the k-error linear complexity, play an important role in cryptology.
Sequences that are suitable as keystreams should possess not only a
large linear complexity but also the change of a few terms must not
cause a significant decrease of the linear complexity.

In this paper we derive the $k$-error linear complexity of binary
sequences of length $2p$  from \cite{DHM} over $\mathbb{F}_p$. These
balanced sequences with optimal three-level autocorrelation are
constructed by cyclotomic classes of order four. Earlier, the linear
complexity and the $k$-error linear complexity over $\mathbb{F}_p$
of the Legendre sequences and series of other cyclotomic sequences
of length $p$ were investigated in \cite{AW,AMW}.

\section{Preliminaries}
\label{sec:1} First, we briefly repeat the basic definitions from
\cite{DHM} and the general information.

Let $p$ be a prime of the form $p\equiv 1 \pmod{4}$, and let
$\theta$ be a primitive root modulo $p$ \cite{IR}. By definition,
put $D_0=\{\theta^{4s} \bmod p; s=1,...,(p-1)/4\}$  and
$D_n=\theta^nD_0, n=1,2,3$. Then these $D_n$ are cyclotomic classes
of order four \cite{H}.

The ring of residue classes $\mathbb{Z}_{2p} \cong
\mathbb{Z}_{2}\times \mathbb{Z}_{p}$ under the isomorphism $\phi
(a)=\left (a \bmod 2, a\bmod p \right )$ \cite{IR}. Ding et al.
considered balanced binary sequences defined as
\begin{equation}
\label{eq1} u_i =\begin{cases}
 1,&\text{ if}  \hspace{4pt} i~\bmod~2p  \in C, \\
  0,&\text{ if}  \hspace{4pt} i~\bmod~2p \not \in C, \\
 \end{cases}
\end{equation}
for $C=\phi^{-1} \left ( \{0\} \times (\{0\}\cup D_m\cup D_j) \cup
\{1\} \times (D_l\cup D_j)\right )$ where $m, j,$ and $l$ are
pairwise distinct integers between 0 and 3 \cite{DHM}. Here we
regard them as  sequences over the finite field $\mathbb{F}_p$.

 By \cite{DHM}, if $\{u_i\}$ has an optimal
autocorrelation value then $p\equiv 5 \pmod{8}$ and $p=1+4y^2$,
$(m,j,l)=(0,1,2), (0,3,2), (1,0,3), (1,2,3)$ or $p=x^2+4, y=-1$,
$(m,j,l)=(0,1,3), (0,2,3), (1,2,0), (1,3,0)$. Here $x, y$ are
integers and $x\equiv 1 (\bmod~4)$.

 It
is well known \cite{CDR} that if $r$ is a binary sequence with
period $N$, then the linear complexity $L(r)$ of this sequence is
defined by
\begin{equation*}
 L(r)=N-\deg  \gcd \bigl (x^N-1,S_r(x)\bigr ),
\end{equation*}
where  $S_r(x) = r_0 + r_1x + ... + r_{N-1}x^{N-1}$. Let's assume we
investigate the linear complexity of $u$ over $\mathbb{F}_p$ and
with a period $2p$. So,
\begin{equation*}
 L(u)=2p-\deg  \gcd \bigl ((x^2-1)^p,S_u(x)\bigr ).
\end{equation*}
The weight of $f(x)$, denoted as $w(f)$, is defined as the number of
nonzero coefficients of $f(x)$.
 From  our definitions it follows
that if the Hamming distance of the vectors $(u_0, u_1, \dots ,
u_{2p-1})$ and $(t_0, t_1, \dots , t_{2p-1})$ is at most $k$ then
there exists $f(x) \in \mathbb{F}_p, \ \ w(f) \leq k$ such that
$S_t(x)=S_u(x)+f(x)$ and the reverse is also true. Therefore
\begin{equation}
\label{eq3}L_k(u)=2p-\max \limits_{f(x)} (m_0+m_1)\end{equation}
where $0\leq m_j\leq p$, $S_u(x)+f(x)\equiv 0 \pmod
{(x-1)^{m_0}(x+1)^{m_1}} $
 and $f(x)
\in \mathbb{F}_p[x],w(f) \leq k$.

Let $g$ be an odd number in the pair $\theta$, $\theta + p$, then
$g$ is a primitive root modulo $2p$ \cite{IR}. By definition, put
$H_0=\{g^{4s}~\bmod~2p; s=1,...,(p-1)/4\}$. Denote by $H_n$ a set
$g^nH_0, n=1,2,3$. Let us introduce the auxiliary polynomial
$S_n(x)=\sum_{i \in H_n} x^i$.
 The following formula was
proved in \cite{EP}.
\begin{equation}
\label{eq4} S_u(x)\equiv (x^p+1)S_{j}(x)+x^pS_{m}(x)+S_{l}(x)+1
\pmod{(x^{2p}-1)}.
\end{equation}
By \eqref{eq4} we have
\begin{equation}
\label{eq5} \begin{cases} S_u(x)\equiv 2S_{j}(x)+S_{m}(x)+S_{l}(x)+1
 \pmod {(x-1)^{p}},\\
 S_u(x)\equiv S_{l}(x)-S_{m}(x)+1
 \pmod {(x+1)^{p}}.
 \end{cases}
\end{equation}
 Let the  sequences $\{q_i\}$ and   $\{v_i\}$ be defined by
\begin{equation}
\label{eq1nn}
 q_i =\begin{cases}
 2,&\text{ if}  \hspace{4pt} i~\bmod~p  \in D_j, \\
 1,&\text{ if}  \hspace{4pt} i~\bmod~p  \in \{0\} \cup D_m \cup D_l, \\
 0,&\text{ otherwise},
 \end{cases}
  \text{  and  }
v_i =\begin{cases}
 1,&\text{ if}  \hspace{4pt} i~\bmod~p  \in \{0\} \cup D_m, \\
  -1,&\text{ if}  \hspace{4pt} i~\bmod~p  \in D_l, \\
 0,&\text{ otherwise}.
 \end{cases}
\end{equation}
 By definition, put $S_q(x)=\sum_{i=0}^{p-1} q_i
x^i$  and $S_v(x)=\sum_{i=0}^{p-1} v_i x^i$. Then by the choice of
$g$ we obtain that
\begin{equation}
\label{eq6} \begin{cases} 2S_{j}(x)+S_{m}(x)+S_{l}(x)+1\equiv
S_{q}(x) \pmod {(x-1)^{p} },\\
S_{m}(x)-S_{l}(x)+1 \equiv S_{v}(x) \pmod {(x-1)^{p} }.
 \end{cases}
\end{equation}

 As noted above, the $k$-error linear
complexity of cyclotomic sequences was investigated in \cite{AMW}.
With the aid of methods from \cite{AMW} it is an easy matter to
prove the following
\begin{equation}
\label{eq7}
 L_k(q) =\begin{cases}
 3(p-1)/4+1,&\text{ if}  \hspace{4pt} 0\leq k \leq (p-1)/4, \\
 (p-1)/2+1,&\text{ if}  \hspace{4pt} (p-1)/4+1 \leq k < (p-1)/3, \\
 1,&\text{ if}  \hspace{4pt}  k =(p-1)/2,\\
  \end{cases}
 \end{equation}
  and
  $(p-1)/4+1\leq L_k(q) \leq (p-1)/2+1$ if $(p-1)/3 \leq k
  <(p-1)/2$.
\begin{equation}
\label{eq8}
 L_k(v) =\begin{cases}
 p,&\text{ if}  \hspace{4pt} k=0, \\
 3(p-1)/4+1,&\text{ if}  \hspace{4pt} 1\leq k < (p-1)/4, \\
 (p-1)/2+1,&\text{ if}  \hspace{4pt} (p-1)/4+1 \leq k < (p-1)/3, \\
 0,&\text{ if}  \hspace{4pt}  k \geq (p-1)/2+1.
 \end{cases}
 \end{equation}
  and
 $9(p-1)/16 \leq L_{(p-1)/4}(v) \leq 3(p-1)/4+1$,  $(p-1)/4\leq L_k(v) \leq (p-1)/2$ if $(p-1)/3 \leq k <(p-1)/2$.

The following statements we also obtain  by \cite {AMW} or by Lemma
3 from \cite{EP}.

\begin{lemma}
\label{l1}
\begin{enumerate}
 \item $S_n(x)= -1/4+(x-1)^{(p-1)/4}E_{n}(x)$ and $E_{n}(1)\neq 0, n=0,1,2,3$;
 \item $S_n(x)= -1/4+(x+1)^{(p-1)/4}F_{n}(x)$ and $F_{n}(-1)\neq 0,
n=0,1,2,3$;
 \item  Let $S_l(x)+S_m(x)+g(x) \equiv 0\pmod{(x-1)^{(p-1)/4+1}}$ and
$|l-m|\neq 2$. Then $w(g(x))\geq (p-1)/4.$
 \end{enumerate}
 \end{lemma}

Let us introduce the auxiliary polynomial
$R(x)=\sum_{i=0}^4c_iS_{i}(x), c_i \in \mathbb{Z}$. Denote a formal
derivative of order $n$ of the polynomial $R(x)$ by $R^{(n)}(x)$.
\begin{lemma}
\label{lnov} Let $R^{(n)}(x)|_{x=\pm 1}=0$ if  $0\leq n \leq
(p-1)/4$. Then $R^{(n)}(x)|_{x=\pm 1}=0$ for $(p-1)/4+1 <n <
(p-1)/2$.
\end{lemma}
\begin{proof} We consider the sequences $\{r_t\}$ of length $p$ defined by
\begin{equation*}r_{t}=\begin{cases}
 0,&\mbox{if  }  t =0, \\
 c_i,&\mbox{if  } t  \in D_i.
 \end{cases}\end{equation*}
 By the definition of the sequence, $S_r(x) \equiv R(x)
 \pmod{(x^p-1)}$, so that by the condition of this lemma
 $L(r) < 3(p-1)/4$. By Theorem 1 from \cite{AMW} for the cyclotomic sequences
  $L(r)=p-c(p-1)/4, 1 \leq c \leq 3$. Hence, $L(r) \leq p-(p-1)/2$. This completes the
 proof of Lemma \ref{lnov}.
\end{proof}
 This lemma can also be proved using Lemma 2 and 3 from \cite{EP}.

\section{The exact values of the $k$-error
linear complexity of $u$ for $1\leq k< (p-1)/4$}  \label{sec:2} In
this section we obtain the upper and lower bounds  of the $k$-error
linear complexity and determine the exact values for the $k$-error
linear complexity $L_k(u), 1\leq k< (p-1)/4$.

First of all, we consider the case $k=1$. Our first contribution in
this paper is the following.
\begin{lemma}
\label{l2}
 Let $\{u_i\}$ be defined by \eqref{eq1} for $p>5$. Then $L_1(u)=(7p+1)/4$.
 \end{lemma}
\begin{proof} Since $L_1(u)\leq L(u)$ and $ L(u)=(7p+1)/4$ \cite{EP}, it
follows that $L_1(u)\leq (7p+1)/4$. Assume that $L_1(u)< L(u)$. Then
there exists $f(x)=ax^b, a\neq 0$ such that $S_u(x) +ax^b \equiv 0
\pmod {(x-1)^{m_0}(x+1)^{m_1}}$ for $m_0+m_1 >(p-1)/4$. By
\eqref{eq5} the last comparison is impossible for $p \neq 5$.
\end{proof}
If $p=5$ then $L_1(u)=8$.
\begin{lemma}
\label{l3}
 Let  $\{u_i\}, \{q_i\}, \{v_i\}$ be defined by \eqref{eq1} and \eqref{eq1nn}, respectively. Then $ L_k(q)+L_k(v)\leq L_k(u)$.
 \end{lemma}
\begin{proof} Suppose $S_u(x)+f(x)\equiv 0 \left(\bmod (x-1)^{m_0}(x+1)^{m_1}
\right )$, $w(f)\leq k$ and $m_0+m_1=2p-L_k(u)$. Combining this with
\eqref{eq5} and \eqref{eq6}  we get $S_{q}(x)+f(x) \equiv 0
 \bigl(\bmod {(x-1)^{m_0}} \bigr )$ and $S_{l}(x)-S_{m}(x)+1+f(x) \equiv 0
 \bigl(\bmod (x+1)^{m_1} \bigr )$ or $S_{m}(x)-S_{l}(x)+1+f(-x) \equiv 0
 \bigl(\bmod (x-1)^{m_1} \bigr )$ Hence $m_0\leq p-L_k(q)$ and $m_1\leq
 p-L_k(v)$. This completes the proof of Lemma \ref{l3}.
 \end{proof}

\begin{lemma}
\label{l4}
 Let  $\{u_i\}$ be defined by \eqref{eq1} and $k\geq 2$. Then $  L_k(u)\leq 3(p-1)/4+1+ L_{k-2}(q)$.
 \end{lemma}
\begin{proof} From our definition it follows that there exists $h(x)$ such that $S_q(x)+h(x)\equiv 0 \pmod {(x-1)^{p-L_{k-2}(q)}}$,
 $w(h)\leq k-2$. Then, by Lemma \ref{l1} $h(x)\equiv 0 \pmod
{(x-1)^{(p-1)/4}}$. Let $h(x)=\sum h_i x^{a_i}$. We consider $f(x)=
\sum f_i x^{b_i}$ where

\noindent
 $
 b_i =\begin{cases}
 a_i,&\text{ if  }  a_i \text{ is an even}, \\
  a_i+p,&\text{ if  }  a_i \text{ is an odd}.
 \end{cases}
$

By definition $f(x)\equiv h(x)\pmod {(x-1)^{p} }$, hence
$S_q(x)+f(x)\equiv 0 \pmod {(x-1)^{p-L_{k-2}(q)}}$. Further, since
$h(x)\equiv 0 \pmod {(x-1)^{(p-1)/4}}$ and $f(x)=f(-x)$, it follows
that $f(x)\equiv 0 \pmod {(x+1)^{(p-1)/4} }$.

Using \eqref{eq4}, we obtain that $S_u(x)+(x^p-1)/2+f(x)\equiv
(x^p-1)\bigl(S_j(x)+S_m(x)+1/2\bigr) +S_q(x)+f(x)\pmod {(x^2-1)^{p}
}.$ From this by Lemma \ref{l1} we can establish that
$S_u(x)+(x^p-1)/2+f(x)\equiv 0 \pmod
{(x-1)^{p-L_{k-2}(q)}(x+1)^{(p-1)/4} }.$  The conclusion of this
lemma then follows from \eqref{eq3}.
 \end{proof}
\begin{theorem}
\label{t6}Let  $\{u_i\}$ be defined by \eqref{eq1} and
 $2\leq k < (p-1)/4 $. Then $L_k(u)=3(p-1)/2+2$.
\end{theorem}
\begin{proof}
By Lemmas \ref{l2} and  \ref{l3} it follows that $L_k(v) +L_k(q)
\leq L_k(u)\leq 3(p-1)/4+1+L_{k-2}(q)$. To conclude the proof, it
remains to note that  $L_k(v) =L_k(q) =L_{k-2}(q)=3(p-1)/4+1$ for
$2\leq k < (p-1)/4 $ by \eqref{eq7},\eqref{eq8}.
 \end{proof}

\section{The estimates of $k$-error
linear complexity }  \label{sec:3}
 In this section we determine the
exact values of the $k$-error linear complexity of $u$ for
$(p-1)/4+2 \leq k <(p-1)/3$ and we obtain the estimates for the
other  values of $k$. Farther, we consider two cases.
\subsection{Let $(m,j,l)=(0,1,3), (0,2,3), (1,2,0)$, $(1,3,0)$}
\label{s1}
\begin{lemma}
\label{l7} Let  $\{u_i\}$ be defined by \eqref{eq1}. Then
$21(p-1)/16+1\leq L_{(p-1)/4}(u)\leq 3(p-1)/2+2$ and  $p+1\leq
L_{(p-1)/4+1}(u)\leq 3(p-1)/2+2$ for $p>5$.
\end{lemma}
The statement of this lemma  follows from Lemmas \ref{l3}, \ref{l4}
and \eqref{eq7}, \eqref{eq8}.
\begin{theorem}
\label{t8}Let  $\{u_i\}$ be defined by \eqref{eq1} for
$(m,j,l)=(0,1,3), (0,2,3), (1,2,0)$, $(1,3,0)$ and  $(p-1)/4+2\leq k
< (p-1)/3 $.  Then  $L_k(u)=p+1$.
\end{theorem}
\begin{proof}
We consider the case when $(m,j,l)=(0,1,3)$. Let
$f(x)=x^p/2-(\rho+3)/4 -(\rho+1)x^pS_0(x)$ where
$\rho=\theta^{(p-1)/4}$ is a primitive $4$-th root of unity modulo
$p$. Then $w(f)=2+(p-1)/4$. Denote $S_u(x)+f(x)$ by $h(x)$. Under
the conditions of this theorem we have
\begin{equation}
\label{eq9} h(x)=(x^p+1)S_1(x)+x^pS_0(x)+ S_3(x)+1 +x^p/2-(\rho+3)/4
-(\rho+1)x^pS_0(x).\end{equation} Hence $h(1)=0$. Let $h^{(n)}(x)$
be a formal derivative of order $n$ of the polynomial $h(x)$. By
Lemmas 2 and 3 from \cite{EP} we have that $h^{(n)}(1)=0$ if $1\leq
n <(p-1)/4$ and by Lemma 3 from \cite{EP}  $h^{(p-1)/4}(1)=\left (
2\rho+1 +\rho^3 -(\rho+1) \right )(p-1)/4=0.$ Hence, by Lemma
\ref{lnov} $h^{(n)}(1)=0$ if $(p-1)/4 < n <(p-1)/2$ and $h(x)\equiv
0 \pmod{(x-1)^{(p-1)/2}}.$

Further, $h(-1)=-1/4+1/4+1-1/2 -(\rho+3)/4+(\rho+1)/4=0$ and
$h^{(p-1)/4}(-1)=\left (-1 +\rho^3 +(\rho+1)\right )(p-1)/4=0.$ So,
by Lemma \ref{lnov} $h^{(n)}(1)=0$ if $1 < n <(p-1)/2$ and
$h(x)\equiv 0\pmod{(x+1)^{(p-1)/2}}.$ Therefore, by \eqref{eq3} we
see that $L_{(p-1)/4+2} \leq p+1$. On the other hand, by Lemma
\ref{l3} $L_k(u)\geq L_k(v)+L_k(q)$. To conclude the proof, it
remains to note that $L_k(v)+L_k(q)= p+1$ for $(p-1)/4 +2< k
<(p-1)/3$ by \eqref{eq7}, \eqref{eq8}.
 The other cases may be considered similarly.
\end{proof}
Farther, if $(p-1)/3 \leq k <(p-1)/2$ then by Lemmas \ref{l3},
Theorem \ref{t8} and \eqref{eq7}, \eqref{eq8} we have that
$(p-1)/2+1 \leq L_k(u)\leq p+1$. It is simple to prove that $
L_{(p-1)/2+2}(u)\leq (p-1)/2+2.$

\subsection{Let $(m,j,l)=(0,1,2), (0,3,2), (1,0,3), (1,2,3)$}
\label{s2} Similarly as in subsection \ref{s1}, we have that
$21(p-1)/16+1\leq L_{(p-1)/4}(u)\leq 3(p-1)/2+2$.
\begin{theorem} \label{t9} Let  $\{u_i\}$ be defined
by \eqref{eq1} for $(m,j,l)=(0,1,2), (0,3,2), (1,0,3),$ $(1,2,3)$
and $(p-1)/4+1\leq k < (p-1)/3 $ then $L_k(u)=5(p-1)/4+2.$
\end{theorem}
\begin{proof} We consider the case when $(m,j,l)=(0,1,2)$. Let here
$f(x)=-1/2-2S_2(x)$ and $h(x)=S_u(x)+f(x)$. Since $(m,j,l)=(0,1,2)$
it follows that
\begin{equation}
\label{eq10} h(x)=(x^p+1)S_1(x)+x^pS_0(x)+
S_2(x)+1-1/2-2S_2(x).\end{equation} Hence $h(1)=0$. By Lemma 2 from
\cite{EP} we have that $h^{(n)}(1)=0$ if $1\leq n <(p-1)/4$. Hence
$h(x)\equiv 0(\bmod~(x-1)^{(p-1)/4}).$

Further, $h(-1)=0$ and $h^{(p-1)/4}(-1)=\left (-1 +\rho^2 -2\rho^2
\right )(p-1)/4=0.$ So, $h^{(n)}(-1)=0$ if $1 < n <(p-1)/2$ and
$h(x)\equiv 0\pmod{(x+1)^{(p-1)/2}}.$ Therefore, by \eqref{eq3} we
see that $L_{(p-1)/4+2} \leq 2p-3(p-1)/4$.

Suppose  $L_{(p-1)/4+2} < 2p-3(p-1)/4$; then by \eqref{eq3} there
exist $m_0, m_1$ such that $m_0+m_1> 3(p-1)/4$ and $S_u(x)+f(x)
\equiv 0
 \pmod {(x-1)^{m_0} (x+1)^{m_1}}$, $w(f)\leq k< (p-1)/3$.

We consider two cases.

(i) Let $m_0\leq (p-1)/4$ or $m_1\leq (p-1)/4$. Then $m_1>(p-1)/2$
or $m_0> (p-1)/2$ and by \eqref{eq5} and \eqref{eq6} we obtain
$L_k(q) < (p+1)/2$ or $L_k(v) < (p+1)/2$. This is impossible for
$k<(p-1)/3$ by \eqref{eq7} or \eqref{eq8}.

(ii) Let $\min(m_0, m_1) > (p-1)/4$. We can write that
$f(x)=f_0(x^2)+xf_1(x^2).$ Therefore,  since
$2S_{1}(x)+S_{0}(x)+S_{2}(x)+1+f(x)\equiv 0 \pmod{ (x-1)^{m_0} }$
and $S_{2}(x)-S_{0}(x)+1+f(x)\equiv 0 \bigl(\bmod (x+1)^{m_1} \bigr
)$ or $-S_{2}(x)+S_{0}(x)+1+f_0(x^2)-xf_1(x^2)\equiv  0 \pmod{
(x-1)^{m_1} }$  we see that $S_{1}(x)+S_{0}(x)+1+f_0(x^2)\equiv 0
\pmod {(x-1)^{\min(m_0,m_1)} }$. Hence, $w(f_0)\geq (p-1)/4$ by
Lemma \ref{l1}.

Similarly, $-2S_{1}(x)-S_{0}(x)-S_{2}(x)+1+f_0(x^2)-xf_1(x^2)\equiv
0 \pmod {(x+1)^{m_1} })$ and
$S_{2}(x)-S_{0}(x)+1+f_0(x^2)+xf_1(x^2)\equiv 0 \bmod {(x+1)^{m_1}
}$ so $S_{1}(x)+S_{2}(x)+1+xf_1(x^2)\equiv 0 \pmod{
(x-1)^{\min(m_0,m_1)} }$. Hence, $w(f_1)\geq (p-1)/4$ by Lemma
\ref{l1}.
 This contradicts  the fact that $w(f)<(p-1)/3.$
\end{proof}
Similarly, if $(p-1)/3 \leq k <(p-1)/2$ then by Lemmas \ref{l3},
Theorem \ref{t8} and \eqref{eq7}, \eqref{eq8} we have that
$(p-1)/2+1 \leq L_k(u)\leq 2p-3(p-1)/4$. Here $ L_{(p-1)/2+2}(u)\leq
3(p-1)/4 +2.$

 In the conclusion of this section note that we can improve the estimate of Lemma \ref{l4} for $k\geq (p-1)/2+1$. With similar arguments as above we obtain the following results
  for $u$.
  \begin{lemma}
\label{l5}
 Let  $\{u_i\}$ be defined by \eqref{eq1} and $k=(p-1)/2+f, f\geq 0$.
 Then $  L_k(u)\leq L_{[f/2]}(v)+1$ where $[f/2]$ is the integral part of number $f/2$.
 \end{lemma}

\section{Conclusion}
We investigated the $k$-error linear complexity over $\mathbb{F}_p$
of
 sequences of length $2p$ with optimal three-level
autocorrelation. These balanced sequences  are constructed by
cyclotomic classes of order four using a method presented by  Ding
et al. We obtained the upper and lower bounds  of $k$-error linear
complexity and determine the exact values of the $k$-error linear
complexity $L_k(u)$ for $1\leq k< (p-1)/4$ and $ (p-1)/4+2\leq k<
(p-1)/3$.

\end{document}